\documentclass[letterpaper, 10 pt, conference]{ieeeconf}  

\IEEEoverridecommandlockouts                              
\usepackage{amsmath,amsfonts,amssymb,color}
\usepackage{graphicx}
\usepackage{psfrag}
\usepackage{algorithm}
\usepackage{algpseudocode}
\usepackage{array}
\usepackage{cite}
\usepackage{footmisc}
\usepackage{mathtools}
\usepackage{enumerate}
\usepackage{soul,color}
\usepackage{caption}
\usepackage{subfig}

\newtheorem{theorem}{Theorem}
\newtheorem{coro}{Corollary}
\newtheorem{lemma}{Lemma}
\newtheorem{remark}{Remark}
\newtheorem{definition}{Definition}
\newtheorem{proposition}{Proposition}
\newtheorem{example}{Example}

\renewcommand{\theenumi}{(\alph{enumi}}

\DeclareMathOperator{\diag}{diag}

\title{\LARGE \bf
Identifying the Dynamics of a System by Leveraging Data \\ from Similar Systems
}

\author{~Lei~Xin, Lintao Ye, George Chiu, Shreyas Sundaram 
\thanks{This research was supported by USDA grant 2018-67007-28439.  This work represents the opinions of the authors and not the USDA or NIFA. Lei Xin and Shreyas Sundaram are with the Elmore Family School of Electrical and Computer Engineering, Purdue University. George Chiu is with the School of Mechanical Engineering, Purdue University. E-mails: {\tt\{lxin, gchiu, sundara2\}@purdue.edu}. Lintao Ye is with the School of Artificial Intelligence and Automation, Huazhong University of Science and Techonology. E-mail: {\tt yelintao93@hust.edu.cn}.}
}

\begin{document}

\maketitle
\thispagestyle{empty}
\pagestyle{empty}

\begin{abstract}
We study the problem of identifying the dynamics of a linear system when one has access to samples generated by a similar (but not identical) system, in addition to data from the true system.  We use a weighted least squares approach and provide finite sample performance guarantees on the quality of the identified dynamics. Our results show that one can effectively use the auxiliary data generated by the similar system to reduce the estimation error due to the process noise, at the cost of adding a portion of error that is due to intrinsic differences in the models of the true and auxiliary systems. We also provide numerical experiments to validate our theoretical results. Our analysis can be applied to a variety of important settings. For example, if the system dynamics change at some point in time (e.g., due to a fault), how should one leverage data from the prior system in order to learn the dynamics of the new system?  As another example, if there is abundant data available from a simulated (but imperfect) model of the true system, how should one weight that data compared to the real data from the system?  Our analysis provides insights into the answers to these questions. 
\end{abstract}


\section{Introduction} \label{sec: introduction}
The problem of dynamical system identification has been an important topic in various fields including economics, control theory and reinforcement learning \cite{ljung1999system}. When modeling from first principles is not possible, one can attempt to learn a predictive model from observed data. While classical system identification techniques focused primarily on achieving asymptotic consistency \cite{bauer1999consistency,jansson1998consistency,knudsen2001consistency}, recent efforts have sought to characterize the number of samples needed to achieve a desired level of accuracy in the learned model. 

The existing literature on finite sample analysis of system identification can be typically divided into two categories: single trajectory-based and multiple trajectories-based. In the single trajectory setup, it is assumed that one has access to a single long record of system input/output data, e.g., \cite{simchowitz2018learning,oymak2019non,simchowitz2019learning,sarkar2019nonparametric,faradonbeh2018finite,sarkar2019near}. This approach enables system identification without restarting the experiment multiple times. In contrast, for the multiple trajectories setup, it is typically assumed that one has access to multiple independent (short) trajectories of the system, e.g., \cite{dean2019sample, fattahi2018data, sun2020finite, zheng2020non, lxinmultiple}. Due to the assumption of independence over multiple trajectories, standard concentration inequalities usually apply in this case. In practice, one may obtain multiple trajectories of the system if one can restart the experiment, or if the data are generated from identical systems running in parallel.

We note that all of the above works assume that the data used for system identification are generated from the true system model that one wants to learn. However, in many cases, collecting abundant data from the true system can be costly or infeasible, and one may want to rely on data generated from similar systems. For example, for non-engineered systems like animals, one may only have a limited amount of data from the true animal one wants to model, due to the challenge of conducting experiments. On the other hand, it may be possible to collect data from other animals in the herd or from a reasonably good simulator, and one may want to leverage these available data. Furthermore, when a system changes its dynamics (e.g., due to failures), one needs to decide whether to discard all of the previous data, or to leverage the old information in an appropriate way. In settings such as the ones described above, it is of great interest to determine how one can leverage the data generated from systems that share similar (but not identical) dynamics. This idea is similar to the notion of {\it transfer learning} in machine learning, where one wants to transfer knowledge from similar tasks to a new task \cite{pan2009survey}. However, in contrast to system identification, most of the papers on transfer learning consider learning a static mapping from a feature space to a label space \cite{bastani2021predicting}.

In this paper, we study the system identification problem using both data generated from the true system that one wants to identify and data generated from an auxiliary (potentially time-varying) system that shares similar dynamics. Similarly to \cite{dean2019sample}, we consider the multiple trajectories setup. However, unlike \cite{dean2019sample}, we use all samples instead of only the last one from each trajectory in order to improve data efficiency. We use a weighted least squares approach and provide a finite sample upper bound on the estimation error. Our result shows that one can leverage the auxiliary data to reduce the error due to the noise, at the cost of adding a bias that depends on the difference between the true and auxiliary systems.  We also provide simulations to validate our theoretical results, and to provide insights into various settings (including the  example scenarios discussed above).

\section{Mathematical notation and terminology} \label{sec: notation and terminology}
Let $\mathbb{R}$ denote the set of real numbers. Let $\lambda_{min}(\cdot)$ and $\lambda_{max}(\cdot)$ be the smallest and largest eigenvalues, respectively, of a symmetric matrix. We use $*$ to denote the conjugate transpose of a given matrix. We use $\|\cdot\|$ and $\|\cdot\|_{F}$ to denote the spectral norm and Frobenius norm of a given matrix, respectively. Vectors are treated as column vectors. A Gaussian distributed random vector is denoted as $u\sim \mathcal{N}(\mu,\Sigma)$, where $\mu$ is the mean and $\Sigma$ is the covariance matrix. We use $I_{n}$ to denote the identity matrix with dimension $n \times n$.

\section{Problem formulation and algorithm} \label{sec: problem formulation}
Consider a discrete time linear time-invariant (LTI) system
\begin{equation}
\begin{aligned} 
\bar{x}_{k+1}=\bar{A}\bar{x}_{k}+\bar{B}\bar{u}_{k}+\bar{w}_{k}, \\
\end{aligned}
\label{eq:True system}
\end{equation}
where $\bar{x}_{k}\in \mathbb{R}^{n}$, $\bar{u}_{k}\in \mathbb{R}^{p}$, $\bar{w}_{k}\in \mathbb{R}^{n}$, are the state, input, and process noise, respectively, and $\bar{A}\in \mathbb{R}^{n \times n}$ and $\bar{B}\in \mathbb{R}^{n \times p}$ are system matrices. 
The input and process noise are assumed to be i.i.d Gaussian, with $\bar{u}_{k} \sim \mathcal{N}(0,\sigma_{\bar{u}}^{2}I_{p})$ and $\bar{w}_{k} \sim \mathcal{N}(0,\sigma_{\bar{w}}^{2}I_{n})$. Note that Gaussian inputs are commonly used in system identification, e.g., \cite{dean2019sample, zheng2020non}. We also assume that both the input $\bar{u}_{k}$ and state $\bar{x}_{k}$ can be perfectly measured.

Suppose that we have access to $N_{r}$ independent experiments of system \eqref{eq:True system}, in which the system restarts from an initial state $\bar{x}_{0} \sim \mathcal{N}(0,\sigma_{\bar{x}}^{2}I_{n})$, and each experiment is of length $T$. We call the state-input pairs collected from each experiment a {\it rollout}, and denote the set of samples we have as $\{(\bar{x}^{i}_{k},\bar{u}^{i}_{k}):1 \leq i \leq N_{r},0 \leq k \leq T\}$, where the superscript denotes the rollout index and the subscript denotes the time index.

Define the matrices 
\begin{equation} \label{G,F}
\begin{aligned}
&\bar{G}_{k}= \begin{bmatrix} \bar{A}^{k-2}\bar{B}& \bar{A}^{k-3}\bar{B}&\cdots& \bar{B}\end{bmatrix}\in \mathbb{R}^{n\times (k-1)p},\\
&\bar{F}_{k}=\begin{bmatrix} \bar{A}^{k-2}& \bar{A}^{k-3}&\cdots& I_{n}\end{bmatrix} \in \mathbb{R}^{n\times(k-1)n},
\end{aligned}
\end{equation}
for $k\geq2$, with $\bar{G}_{1}=0$ and $\bar{F}_{1}=0$. For $k\geq2$, we have
\begin{equation}
\begin{aligned}
\bar{x}^{i}_{k-1}=\bar{G}_{k}
\begin{bmatrix}
\bar{u}^{i}_{0}\\
\vdots\\
\bar{u}^{i}_{k-2}
\end{bmatrix}+
\bar{F}_{k}
\begin{bmatrix}
\bar{w}^{i}_{0}\\
\vdots\\
\bar{w}^{i}_{k-2}
\end{bmatrix}+
\bar{A}^{k-1}\bar{x}^{i}_{0}.
\end{aligned}
\end{equation}

Letting $\bar{z}^{i}_{k}=\begin{bmatrix} \bar{x}^{i*}_{k}& \bar{u}^{i*}_{k} \end{bmatrix}^{*}\in \mathbb{R}^{n+p}$ for $k\geq 0$, one can verify that $\bar{z}^{i}_{k}\sim \mathcal{N}(0,
\bar{\Sigma}_{k})$,
where
\begin{equation} \label{covariance}
\bar{\Sigma}_{k}=\left[\begin{smallmatrix}
\sigma_{\bar{u}}^{2}\bar{G}_{k+1}\bar{G}_{k+1}^{*}+\sigma_{\bar{w}}^{2}\bar{F}_{k+1}\bar{F}_{k+1}^{*}+\sigma_{\bar{x}}^{2}\bar{A}^{k}\bar{A}^{k*}&0\\
0&\sigma_{\bar{u}}^{2}I_{p}
\end{smallmatrix}\right].
\end{equation}
Note that we have $\lambda_{min}(\bar{\Sigma}_{k})> 0$ for all $k \geq 0$.

For each rollout $i$, define $\bar{X}^{i}=\begin{bmatrix}
\bar{x}^{i}_{T}&\cdots&\bar{x}^{i}_{1}
\end{bmatrix}  \in \mathbb{R}^{n\times T}$, $\bar{Z}^{i}=\begin{bmatrix}
\bar{z}^{i}_{T-1}&\cdots&\bar{z}^{i}_{0}
\end{bmatrix} \in \mathbb{R}^{(n+p)\times T}$, $\bar{W}^{i}=\begin{bmatrix}
\bar{w}^{i}_{T-1}&\cdots&\bar{w}^{i}_{0}
\end{bmatrix}  \in \mathbb{R}^{n\times T}$.
Further, define the batch matrices $\bar{X}=\begin{bmatrix}\bar{X}^{1}&\cdots&\bar{X}^{N_{r}}\end{bmatrix}\in \mathbb{R}^{n\times N_{r}T},\bar{Z}=\begin{bmatrix}\bar{Z}^{1}&\cdots&\bar{Z}^{N_{r}}\end{bmatrix}\in \mathbb{R}^{(n+p)\times N_{r}T},\bar{W}=\begin{bmatrix}\bar{W}^{1}&\cdots&\bar{W}^{N_{r}}\end{bmatrix}\in \mathbb{R}^{n\times N_{r}T}$. Denoting $\Theta \triangleq \begin{bmatrix}
\bar{A}&\bar{B}\end{bmatrix}\in \mathbb{R}^{n\times(n+p)}$, we have
\begin{equation*}
\begin{aligned}
&\bar{X}=\Theta\bar{Z}+\bar{W}.
\end{aligned}
\end{equation*}

In general, one would like to solve:
\begin{equation*} 
\begin{aligned}
    \mathop{\min}_{\tilde{\Theta}\in \mathbb{R}^{n\times (n+p)}} \|\bar{X}-\tilde{\Theta}\bar{Z}\|^{2}_{F},
\end{aligned}
\end{equation*}
and obtain an estimate $\Theta_{LS} \triangleq \begin{bmatrix}\bar{A}_{LS}& \bar{B}_{LS}\end{bmatrix}$, of which the analytical form is
\begin{equation*} 
\begin{aligned}
\Theta_{LS}=\bar{X}\bar{Z}^{*}(\bar{Z}\bar{Z}^{*})^{-1},
\end{aligned}
\end{equation*}
assuming invertibility of the matrix $\bar{Z}\bar{Z}^{*}$.
However, without enough samples (i.e., if $N_{r}$ is small, and there is no single long run record available), the obtained estimate could have large estimation error. In such cases, we can rely on samples generated from an auxiliary system, if it shares ``similar'' dynamics to system \eqref{eq:True system}. In particular, consider an auxiliary discrete time linear time-varying system
\begin{equation}
\begin{aligned} 
\hat{x}_{k+1}=\hat{A}_{k}\hat{x}_{k}+\hat{B}_{k}\hat{u}_{k}+\hat{w}_{k},\\
\end{aligned}
\label{eq:Perturbed system}
\end{equation}
where $\hat{x}_{k}\in \mathbb{R}^{n}$, $\hat{u}_{k}\in \mathbb{R}^{p}$, $\hat{w}_{k}\in \mathbb{R}^{n}$ are the state, input, and process noise, respectively, and $\hat{A}_{k}\in \mathbb{R}^{n \times n}$ and $\hat{B}_{k}\in \mathbb{R}^{n \times p}$ are system matrices. Again, the input and process noise are assumed to be i.i.d Gaussian, with $\hat{u}_{k} \sim \mathcal{N}(0,\sigma_{\hat{u}}^{2}I_{p})$ and $\hat{w}_{k} \sim \mathcal{N}(0,\sigma_{\hat{w}}^{2}I_{n})$. Note that an LTI system is a special case of the above system, with $\hat{A}_{k}=\hat{A}$ and  $\hat{B}_{k}=\hat{B}$ for all $k\ge0$.  
The dynamics of system \eqref{eq:Perturbed system} can be rewritten as
\begin{equation} 
\begin{aligned}
\hat{x}_{k+1}=(\bar{A}+\delta_{A_{k}})\hat{x}_{k}+(\bar{B}+\delta_{B_{k}})\hat{u}_{k}+\hat{w}_{k},\\
\end{aligned}
\end{equation}
where $\delta_{A_{k}}=\hat{A}_{k}-\bar{A},\delta_{B_{k}}=\hat{B}_{k}-\bar{B}$. Intuitively, the samples generated from the above system will be useful for identifying system \eqref{eq:True system} if $\|\delta_{A_{k}}\|,\|\delta_{B_{k}}\|$ are small for all $k$.

Now, suppose that we also have access to $N_{p}$ independent experiments of system \eqref{eq:Perturbed system}, in which the system restarts from an initial state $\hat{x}_{0}\sim \mathcal{N}(0,\sigma_{\hat{x}}^{2})$, and each experiment is of length $T$. Let $\{(\hat{x}^{i}_{k},\hat{u}^{i}_{k}):1 \leq i \leq N_{p},0 \leq k \leq T\}$ denote the samples from these experiments. Define $\Phi(k,l)=\hat{A}_{k-1}\hat{A}_{k-2}\cdots\hat{A}_{l}$ for $k>l$ with $\Phi(k,l)=I_{n}$ when $k=l$. Let $\hat{G}_{k} \in \mathbb{R}^{n\times (k-1)p}$ and $\hat{F}_{k}\in \mathbb{R}^{n\times (k-1)n}$ be
\begin{equation}
\begin{aligned}
&\hat{G}_{k}=\begin{bmatrix}\Phi(k-1,1)\hat{B}_{0}& \Phi(k-1,2)\hat{B}_{1}&\cdots& \hat{B}_{k-2} \end{bmatrix}, \\
&\hat{F}_{k}=\begin{bmatrix}\Phi(k-1,1)& \Phi(k-1,2)&\cdots& I_{n} \end{bmatrix}.
\end{aligned}
\end{equation}
for $k\geq2$, with $\hat{G}_{1}=0$ and $\hat{F}_{1}=0$. For $k\geq2$. we have
\begin{equation}
\begin{aligned}
\hat{x}^{i}_{k-1}=\hat{G}_{k}
\begin{bmatrix}
\hat{u}^{i}_{0}\\
\vdots\\
\hat{u}^{i}_{k-2}
\end{bmatrix}+
\hat{F}_{k}
\begin{bmatrix}
\hat{w}^{i}_{0}\\
\vdots\\
\hat{w}^{i}_{k-2}
\end{bmatrix}+
\Phi(k-1,0)\hat{x}^{i}_{0}.
\end{aligned}
\end{equation}

Letting $\hat{z}^{i}_{k}=\begin{bmatrix} \hat{x}^{i*}_{k}& \hat{u}^{i*}_{k} \end{bmatrix}^{*}\in \mathbb{R}^{n+p}$ for $k\geq 0$, one can verify that $\hat{z}^{i}_{k}\sim \mathcal{N}(0,
\hat{\Sigma}_{k})$, where
\begin{equation} \label{covariance_p}
\hat{\Sigma}_{k}=\left[\begin{smallmatrix}
\sigma_{\hat{u}}^{2}\hat{G}_{k+1}\hat{G}_{k+1}^{*}+\sigma_{\hat{w}}^{2}\hat{F}_{k+1}\hat{F}_{k+1}^{*}+\sigma_{\hat{x}}^{2}\Phi(k,0)\Phi(k,0)^{*} & 0\\
0 &\sigma_{\hat{u}}^{2}I_{p}
\end{smallmatrix}\right].
\end{equation}
Again, we have $\lambda_{min}(\hat{\Sigma}_{k})>0$ for all $k \geq 0$.

Further, the matrices  $\hat{X}^{i}\in \mathbb{R}^{n\times T},\hat{Z}^{i}\in \mathbb{R}^{(n+p)\times T},\hat{W}^{i}\in \mathbb{R}^{n\times T},\hat{X}\in \mathbb{R}^{n\times N_{p}T},\hat{Z}\in \mathbb{R}^{(n+p)\times N_{p}T},\hat{W}\in \mathbb{R}^{n\times N_{p}T}$ are defined similarly, using $\hat{u}^{i}_{k},\hat{x}^{i}_{k},\hat{w}^{i}_{k}$ from system \eqref{eq:Perturbed system}.  Let $X=\begin{bmatrix}\bar{X}&\hat{X}\end{bmatrix}\in \mathbb{R}^{n\times (N_{r}+N_{p})T}, Z=\begin{bmatrix}\bar{Z}&\hat{Z}\end{bmatrix}\in \mathbb{R}^{(n+p)\times (N_{r}+N_{p})T}, W=\begin{bmatrix}\bar{W}&\hat{W}\end{bmatrix}\in \mathbb{R}^{n\times (N_{r}+N_{p})T}$ and $\delta_{\Theta_{k}}=\begin{bmatrix}
\delta_{A_{k}}&\delta_{B_{k}}\end{bmatrix}\in \mathbb{R}^{n\times (n+p)}$. Defining 
\begin{equation*}
\Delta^{i}=\begin{bmatrix}
\delta_{\Theta_{T-1}} \hat{z}^{i}_{T-1}&\cdots&\delta_{\Theta_{0}}\hat{z}^{i}_{0}
\end{bmatrix}\in \mathbb{R}^{n\times T},\\
\end{equation*}
for all $i\in\{1,\dots,N_p\}$, and denoting
\begin{equation*}
\Delta=\begin{bmatrix}
0&\cdots&0&\Delta^{1} &\cdots& \Delta^{N_{p}}\end{bmatrix}\in \mathbb{R}^{n\times (N_{r}+N_{p})T},
\end{equation*}
we have the relationship
\begin{equation} \label{relationship}
\begin{aligned}
&X=\Theta Z+W+\Delta.
\end{aligned}
\end{equation}
Letting $q_{k} \in \mathbb{R}_{\ge 0}$ be a design parameter that specifies the relative weight assigned to samples generated from the auxiliary system \eqref{eq:Perturbed system} at time step $k$, we can define $\mathcal{Q}=\diag(q_{T-1},\cdots,q_{0})\in \mathbb{R}^{T\times T}$ and $\hat{Q}=\diag(\mathcal{Q},\cdots,\mathcal{Q})\in \mathbb{R}^{N_{p}T\times N_{p}T}$. Further, define $Q=\diag(I_{N_{r}T},\hat{Q})\in \mathbb{R}^{(N_{p}+N_{r})T\times (N_{p}+N_{r})T}$. We are interested in the following weighted least squares problem:
\begin{equation} \label{Weighted pb}
\begin{aligned}
    \mathop{\min}_{\tilde{\Theta}\in \mathbb{R}^{n\times (n+p)}} \|XQ^{\frac{1}{2}}-\tilde{\Theta}ZQ^{\frac{1}{2}}\|^{2}_{F}.
\end{aligned}
\end{equation}
The well known weighted least squares estimate is $\Theta_{WLS} \triangleq\begin{bmatrix}\bar{A}_{WLS}&\bar{B}_{WLS}\end{bmatrix}$, which has the form
\begin{equation} 
\begin{aligned}
\Theta_{WLS}=XQZ^{*}(ZQZ^{*})^{-1},
\end{aligned}
\end{equation}
when the matrix $ZQZ^{*}$ is invertible. Using \eqref{relationship}, the estimation error can be expressed as
\begin{equation}
\Theta_{WLS}-\Theta=WQZ^{*}(ZQZ^{*})^{-1}+\Delta QZ^{*}(ZQZ^{*})^{-1} .\label{error_W}
\end{equation}


The above steps are encapsulated in Algorithm \ref{BatchAlgorithm}.
\begin{algorithm}[H] 
\caption{System Identification Using Auxiliary Data} \label{BatchAlgorithm}
\begin{algorithmic}[1]
\State Gather $N_{r}$ length $T$ rollouts of samples generated from the true system \eqref{eq:True system}, where $\bar{x}_{0}^{i} \sim \mathcal{N}(0,\sigma_{\bar{x}}^{2}I_{n})$ for all $1 \leq i \leq N_{r}$.
\State Gather $N_{p}$ length $T$ rollouts of samples generated from the auxiliary system  \eqref{eq:Perturbed system}, where $\hat{x}_{0}^{i}\sim \mathcal{N}(0,\sigma_{\hat{x}}^{2}I_{n})$ for all $1 \leq i \leq N_{p}$.
\State Construct the matrices $X,Q,Z$. Compute $\Theta_{WLS}=XQZ^{*}(ZQZ^{*})^{-1}$.
\State Return the first $n$ columns of $\Theta_{WLS}$ as an estimated $\bar{A}$, and the remaining columns of $\Theta_{WLS}$ as an estimated $\bar{B}$.
\end{algorithmic}
\end{algorithm}

\begin{remark}
Note that the weight parameter $q_{k}$ specifies how much we should weight the data from the auxiliary system relative to the data from the true system, and can be a function of the number of rollouts ($N_r$ and $N_p$) from each of those systems.  Our analysis in the next section, and subsequent evaluation, will provide guidance on the appropriate choice of weight parameter.
\end{remark}


Next, we provide upper bounds on the estimation error $\|\Theta_{WLS}-\Theta\|$ as a function of $N_{r}, N_{p}, \|\delta_{\Theta_{k}}\|, q_{k}$ and other system parameters. Moreover, we will provide insights into the case when the auxiliary system is time-invariant.

\section{Analysis of the System Identification Error}
To upper bound the estimation error in \eqref{error_W}, we will upper bound the error terms $\|WQZ\|, \|(ZQZ^{*})^{-1}\|$, and $\|\Delta QZ^{*}\|$ separately.  We will start with some intermediate results pertaining to these quantities. 

\subsection{Intermediate Results}
We will rely on the following lemma from \cite[Corollary~5.35]{vershynin2010introduction}, which provides non-asymptotic lower bound and upper bound of a standard Wishart matrix.
\begin{lemma} Let $u_{i}\sim \mathcal{N}(0,I_{n+p})$, $i=1,\ldots,N$ be i.i.d random vectors. For any fixed $\delta >0$, with probability at least $1-\delta$, we have the following inequalities:
\begin{equation*}
   \sqrt{\lambda_{min}(\sum_{i=1}^{N}u_{i}u_{i}^{*}})\geq \sqrt{N}-\sqrt{n+p}-\sqrt{2\log{\frac{2}{\delta}}},
\end{equation*}
\begin{equation*}
   \sqrt{\lambda_{max}(\sum_{i=1}^{N}u_{i}u_{i}^{*}})\leq \sqrt{N}+\sqrt{n+p}+\sqrt{2\log{\frac{2}{\delta}}}.
\end{equation*}
\label{lemma:Bound of unit variance gaussian}
\end{lemma}

We have the following results.
\begin{proposition} \label{Prop:Lower Bound Aux}
For any fixed $ \delta >0$, let $N_{p}\geq N_{0} \triangleq 8(n+p)+16\log\frac{2T}{\delta}$. With probability at least $1-\delta$, we have the following inequalities:
\begin{equation*}
\hat{Z}\hat{Q}\hat{Z}^{*} \succeq \frac{N_{p}}{4}\sum_{k=0}^{T-1}q_{k}\hat{\Sigma}_{k},
\end{equation*}
\begin{equation*}
\|\Delta QZ^{*}\|\leq\frac{9N_{p}}{4}\sum_{k=0}^{T-1}q_{k}\|\delta_{\Theta_{k}}\|\|\hat{\Sigma}_{k}\|.
\end{equation*}
\end{proposition}
\begin{proof}
We have
\begin{align*}
\hat{Z}\hat{Q}\hat{Z}^{*}=\sum_{i=1}^{N_{p}}\hat{Z}^{i}\mathcal{Q}\hat{Z}^{i*}
 =\sum_{i=1}^{N_{p}}\sum_{k=0}^{T-1}q_{k}\hat{z}^{i}_{k}\hat{z}^{i*}_{k}=\sum_{k=0}^{T-1}\sum_{i=1}^{N_{p}}q_{k}\hat{z}^{i}_{k}\hat{z}^{i*}_{k}
 .
\end{align*}
For any fixed $k$, define $\hat{\mathbf{u}}^{i}_{k}=\hat{\Sigma}_{k}^{-\frac{1}{2}}\hat{z}^{i}_{k}$. Note that $\hat{\mathbf{u}}^{i}_{k}$ are i.i.d random vectors with $\hat{\mathbf{u}}^{i}_{k}\sim \mathcal{N}(0,I_{n+p})$ for $i\in\{1,2,\ldots,N_{p}\}$. Hence, the above sum can be written as
\begin{equation} \label{sum of z lower}
\begin{aligned}
&\sum_{k=0}^{T-1}\sum_{i=1}^{N_{p}}q_{k}\hat{\Sigma}_{k}^{\frac{1}{2}}\hat{\mathbf{u}}^{i}_{k}\hat{\mathbf{u}}^{i*}_{k}\hat{\Sigma}_{k}^{\frac{1}{2}}=\sum_{k=0}^{T-1}q_{k}\hat{\Sigma}_{k}^{\frac{1}{2}}(\sum_{i=1}^{N_{p}}\hat{\mathbf{u}}^{i}_{k}\hat{\mathbf{u}}^{i*}_{k})\hat{\Sigma}_{k}^{\frac{1}{2}}.\\
\end{aligned}
\end{equation}

Similarly, we have
\begin{align}\nonumber
\|\Delta QZ^{*}\|&=\|\sum_{i=1}^{N_{p}}\Delta^{i}\mathcal{Q}\hat{Z}^{i*}\|=\|\sum_{i=1}^{N_{p}}\sum_{k=0}^{T-1}q_{k}\delta_{\Theta_{k}}\hat{z}^{i}_{k}\hat{z}^{i*}_{k}\|\\\nonumber
&=\|\sum_{k=0}^{T-1}\sum_{i=1}^{N_{p}}q_{k}\delta_{\Theta_{k}}\hat{z}^{i}_{k}\hat{z}^{i*}_{k}\|\\\nonumber
&\leq \sum_{k=0}^{T-1}q_{k}\|\delta_{\Theta_{k}}\|\|\sum_{i=1}^{N_{p}}\hat{z}^{i}_{k}\hat{z}^{i*}_{k}\|\\\nonumber
&=\sum_{k=0}^{T-1}q_{k}\|\delta_{\Theta_{k}}\|\|\hat{\Sigma}_{k}^{\frac{1}{2}}(\sum_{i=1}^{N_{p}}\hat{\mathbf{u}}^{i}_{k}\hat{\mathbf{u}}^{i*}_{k})\hat{\Sigma}_{k}^{\frac{1}{2}}\|\\
&\leq \sum_{k=0}^{T-1}q_{k}\|\delta_{\Theta_{k}}\|\|\hat{\Sigma}_{k}\|\|\sum_{i=1}^{N_{p}}\hat{\mathbf{u}}^{i}_{k}\hat{\mathbf{u}}^{i*}_{k}\|.\label{sum of z upper}
\end{align}

Fixing $\delta>0$ and $k$ and applying Lemma \ref{lemma:Bound of unit variance gaussian}, we have with probability at least $1-\frac{\delta}{T}$ the following:
\begin{equation*}
   \sqrt{\lambda_{min}(\sum_{i=1}^{N_{p}}\mathbf{u}^{i}_{k}\mathbf{u}^{i*}_{k}})\geq \sqrt{N_{p}}-\sqrt{n+p}-\sqrt{2\log{\frac{2T}{\delta}}},
\end{equation*} 
\begin{equation*}
   \sqrt{\lambda_{max}(\sum_{i=1}^{N_{p}}\mathbf{u}^{i}_{k}\mathbf{u}^{i*}_{k}})\leq \sqrt{N_{p}}+\sqrt{n+p}+\sqrt{2\log{\frac{2T}{\delta}}}.
\end{equation*}
Further, we have
\begin{equation*}
\begin{aligned}
  &\frac{1}{2}\sqrt{N_{p}}\geq \sqrt{n+p}+\sqrt{2\log{\frac{2T}{\delta}}}\\
   \Longleftrightarrow\quad &\frac{N_{p}}{4}\geq(\sqrt{n+p}+\sqrt{2\log{\frac{2T}{\delta}}})^{2}.
   \end{aligned}
\end{equation*}
Noting the inequality $2(a^2+b^2)\geq(a+b)^{2}$, we can write
\begin{equation*}
\begin{aligned}
 2(n+p+2\log{\frac{2T}{\delta}})\geq  (\sqrt{n+p}+\sqrt{2\log{\frac{2T}{\delta}}})^2.
   \end{aligned}
\end{equation*}
Letting $N_{p}\geq 8(n+p)+16\log{\frac{2T}{\delta}}$, one can then show that the following inequalities hold with probability at least $1-\frac{\delta}{T}$:
\begin{equation*}
\begin{aligned}
 \sqrt{\lambda_{min}(\sum_{i=1}^{N_{p}}\mathbf{u}^{i}_{k}\mathbf{u}^{i*}_{k}})&\geq \frac{1}{2}\sqrt{N_{p}}+\frac{1}{2}\sqrt{N_{p}}\\
 &\quad-\sqrt{n+p}-\sqrt{2\log{\frac{2T}{\delta}}}\\
   &\geq\frac{1}{2}\sqrt{N_{p}},
   \end{aligned}
\end{equation*}
\begin{equation*}
\begin{aligned}
 \sqrt{\lambda_{max}(\sum_{i=1}^{N_{p}}\mathbf{u}^{i}_{k}\mathbf{u}^{i*}_{k}})&\leq \sqrt{N_{p}}+\sqrt{n+p}+\sqrt{2\log{\frac{2T}{\delta}}}\\
   &\leq\frac{3}{2}\sqrt{N_{p}}.
   \end{aligned}
\end{equation*}
It follows that with probability at least $1-\frac{\delta}{T}$,
\begin{equation*}
\begin{aligned}
q_{k}\hat{\Sigma}_{k}^{\frac{1}{2}}(\sum_{i=1}^{N_{p}}\hat{\mathbf{u}}^{i}_{k}\hat{\mathbf{u}}^{i*}_{k})\hat{\Sigma}_{k}^{\frac{1}{2}}\succeq q_{k}\frac{N_{p}}{4}\hat{\Sigma}_{k},
\end{aligned}
\end{equation*}
and
\begin{equation*}
\begin{aligned}
q_{k}\|\delta_{\Theta_{k}}\|\|\hat{\Sigma}_{k}\|\|\sum_{i=1}^{N_{p}}\hat{\mathbf{u}}^{i}_{k}\hat{\mathbf{u}}^{i*}_{k}\|\leq\frac{9N_{p}}{4} q_{k}\|\delta_{\Theta_{k}}\|\|\hat{\Sigma}_{k}\|.
\end{aligned}
\end{equation*}

The result follows by applying a union bound for all $k$ in \eqref{sum of z lower} and \eqref{sum of z upper}, 
\end{proof}
\begin{proposition} \label{Prop:Lower Bound True}
Fixing $\delta>0$, and letting $N_{r}\geq N_{0}=8(n+p)+16\log\frac{2T}{\delta}$, we have with probability at least $1-\delta$, 
\begin{equation*}
\bar{Z}\bar{Z}^{*} \succeq \frac{N_{r}}{4}\sum_{k=0}^{T-1}\bar{\Sigma}_{k}.
\end{equation*}
\end{proposition}
\begin{proof}
Replacing $q_{k}$ by $1$, the proof follows directly from Proposition \ref{Prop:Lower Bound Aux}.
\end{proof}

Next, We will leverage the following lemma from\cite[Lemma~1]{dean2019sample} to bound the contribution from the noise terms.
\begin{lemma}
Let $f_{i}\in\mathbb{R}^{m}$, $g_{i}\in\mathbb{R}^{n}$ be independent random vectors $f_{i}\sim \mathcal{N}(0,\Sigma_{f})$ and $g_{i}\sim \mathcal{N}(0,\Sigma_{g})$, for $i=1,\cdots,N$. Let $N\geq2(n+m)\log{\frac{1}{\delta}}$. For any fixed $ \delta>0$, we have with probability at least $1-\delta$,
\begin{equation*}
   \|\sum_{i=1}^{N}f_{i}g_{i}^{*}\|\leq4\|\Sigma_{f}\|^{\frac{1}{2}}\|\Sigma_{g}\|^{\frac{1}{2}}\sqrt{N(m+n)\log{\frac{9}{\delta}}}.
\end{equation*}
\label{lemma:upper bound two independent gaussian}
\end{lemma}

We have the following results.
\begin{proposition} \label{Prop:Upper Bound Aux}
For any fixed $ \delta >0$, let $N_{p}\geq N_{1} \triangleq (4n+2p)\log\frac{T}{\delta}$. We have with probability at least $1-\delta$
\begin{equation*}
\|\hat{W}\hat{Q}\hat{Z}^{*}\|\leq4\sigma_{\hat{w}}\sqrt{N_{p}(2n+p)\log\frac{9T}{\delta}}\sum_{k=0}^{T-1}q_{k}\|\hat{\Sigma}_{k}^{\frac{1}{2}}\|.
\end{equation*}
\end{proposition}
\begin{proof}
From the definitions of $\hat{W}$, $\hat{Q}$ and $\hat{Z}$, 
\begin{equation} \label{to union}
\begin{aligned}
\|\hat{W}\hat{Q}\hat{Z}^{*}\|&=\|\sum_{i=1}^{N_{p}}\hat{W}^{i}\mathcal{Q}\hat{Z}^{i*}\|
 =\|\sum_{i=1}^{N_{p}}\sum_{k=0}^{T-1}q_{k}\hat{w}^{i}_{k}\hat{z}^{i*}_{k}\|\\
 &=\|\sum_{k=0}^{T-1}\sum_{i=1}^{N_{p}}q_{k}\hat{w}^{i}_{k}\hat{z}^{i*}_{k}\|\leq\sum_{k=0}^{T-1}\|\sum_{i=1}^{N_{p}}q_{k}\hat{w}^{i}_{k}\hat{z}^{i*}_{k}\|.\\
\end{aligned}
\end{equation}

Fix $\delta>0$, and let $N_{p}\geq(4n+2p)\log\frac{T}{\delta}$. For any fixed $k$, we can apply Lemma \ref{lemma:upper bound two independent gaussian} to obtain that with probability at least $1-\frac{\delta}{T}$,
\begin{equation*}
\begin{aligned}
\|\sum_{i=1}^{N_{p}}q_{k}\hat{w}^{i}_{k}\hat{z}^{i*}_{k}\|\leq4q_{k}\sigma_{\hat{w}}\|\hat{\Sigma}_{k}^{\frac{1}{2}}\|\sqrt{N_{p}(2n+p)\log\frac{9T}{\delta}}.\\
\end{aligned}
\end{equation*}
Further applying a union bound for all $k$ in  \eqref{to union}, we get the desired result.
\end{proof}

\begin{proposition} \label{Prop:Upper Bound True}
For any fixed $ \delta >0$, let $N_{r}\geq N_{1}=(4n+2p)\log\frac{T}{\delta}$. We have with probability at least $1-\delta$,
\begin{equation*}
\|\bar{W}\bar{Z}^{*}\|\leq4\sigma_{\bar{w}}\sqrt{N_{r}(2n+p)\log\frac{9T}{\delta}}\sum_{k=0}^{T-1}\|\bar{\Sigma}_{k}^{\frac{1}{2}}\|.
\end{equation*}
\end{proposition}
\begin{proof}
Replacing $q_{k}$ by $1$, the proof follows directly from Proposition \ref{Prop:Upper Bound Aux}.
\end{proof}


\subsection{Main Results and Insights}
We now come to the main result of our paper, providing a bound on the system identification error in \eqref{error_W}.

\begin{theorem} \label{Bound Overall}
For any fixed $\delta>0$, let $\min{(N_{r},N_{p}})\geq \max\{N_{0},N_{1}\}$, where $N_{0}=8(n+p)+16\log\frac{2T}{\delta}, N_{1} =(4n+2p)\log\frac{T}{\delta}$. Then, with probability at least $1-4\delta$, the least squares solution to problem \eqref{Weighted pb} satisfies
\begin{equation*}
\begin{aligned}
&\max{(\|\bar{A}_{WLS}-\bar{A}\|,\|\bar{B}_{WLS}-\bar{B}\|)}\\ \le&\underbrace{\frac{c_{0}(\sqrt{N_{r}}\sigma_{\bar{w}}\sum_{k=0}^{T-1}\|\bar{\Sigma}_{k}^{\frac{1}{2}}\|+\sqrt{N_{p}}\sigma_{\hat{w}}\sum_{k=0}^{T-1}q_{k}\|\hat{\Sigma}_{k}^{\frac{1}{2}}\|)}{\lambda_{min}(N_{r}\sum_{k=0}^{T-1}\bar{\Sigma}_{k}+N_{p}\sum_{k=0}^{T-1}q_{k}\hat{\Sigma}_{k})}}_\text{Error due to noise}\\ 
&+\underbrace{\frac{9N_{p}\sum_{k=0}^{T-1}q_{k}\|\delta_{_{\Theta_{k}}}\|\|\hat{\Sigma}_{k}\|}{\lambda_{min}(N_{r}\sum_{k=0}^{T-1}\bar{\Sigma}_{k}+N_{p}\sum_{k=0}^{T-1}q_{k}\hat{\Sigma}_{k})}}_\text{Error due to difference between true and auxiliary systems},
\end{aligned}
\end{equation*}
where
\begin{equation*}
\begin{aligned}
c_{0}=16\sqrt{(2n+p)\log\frac{9T}{\delta}}.
\end{aligned}
\end{equation*}
\end{theorem}
\begin{proof}
Recall that the estimation error in \eqref{error_W} satisfies $\|\Theta_{WLS}-\Theta\|\le\|WQZ^{*}\|\|(ZQZ^{*})^{-1}\|+\|\Delta QZ^{*}\|\|(ZQZ^{*})^{-1}\|$. 
First, note that $ZQZ^{*}=\bar{Z}\bar{Z}^{*}+\hat{Z}\hat{Q}\hat{Z}^{*}$. Fixing $\delta>0$, letting $\min{(N_{r},N_{p}})\geq N_{0}$, combining Proposition \ref{Prop:Lower Bound True} and the lower bound in Proposition \ref{Prop:Lower Bound Aux} using a union bound, and taking the inverse, we obtain that with probability at least $1-2\delta$,
\begin{equation}\label{1st}
\|(ZQZ^{*})^{-1}\|\leq\frac{4}{\lambda_{min}(N_{r}\sum_{k=0}^{T-1}\bar{\Sigma}_{k}+N_{p}\sum_{k=0}^{T-1}q_{k}\hat{\Sigma}_{k})}.
\end{equation}

Similarly, note that $WQZ^{*}=\bar{W}\bar{Z}^{*}+\hat{W}\hat{Q}\hat{Z}^{*}$. Letting $\min{(N_{r},N_{p}})\geq N_{1}$, and combining Proposition \ref{Prop:Upper Bound True} and Proposition \ref{Prop:Upper Bound Aux} using a union bound, we obtain that with probability at least $1-2\delta$,
\begin{align}\nonumber
&\|WQZ^{*}\|\le4\sigma_{\bar{w}}\sqrt{N_{r}(2n+p)\log\frac{9T}{\delta}}\sum_{k=0}^{T-1}\|\bar{\Sigma}_{k}^{\frac{1}{2}}\|\\
&\qquad\qquad +4\sigma_{\hat{w}}\sqrt{N_{p}(2n+p)\log\frac{9T}{\delta}}\sum_{k=0}^{T-1}q_{k}\|\hat{\Sigma}_{k}^{\frac{1}{2}}\|.\label{2nd}
\end{align}

Finally, combining \eqref{1st}-\eqref{2nd} via a union bound, using the upper bound in Proposition \ref{lemma:Bound of unit variance gaussian}, and using the fact that $\bar{A}_{WLS}-\bar{A},\bar{B}_{WLS}-\bar{B}$ are submatrices of $\Theta_{WLS}-\Theta$, we get the desired result.
\end{proof}

Note that the first term in the bound on the error corresponds to the error due to the noise from the two systems, and the second term corresponds to the error due to the intrinsic model difference. Now, we state a corollary of the above theorem when the auxiliary system is time-invariant, which helps us gain more insights into the above result.

\begin{coro}\label{Bound LTI}
Suppose that the auxiliary system \eqref{eq:Perturbed system} is time-invariant, i.e., $\hat{A}_{k}=\hat{A}, \hat{B}_{k}=\hat{B}, \delta_{\Theta_{k}}=\delta_{\Theta}=\begin{bmatrix}
\delta_{A}&\delta_{B}\end{bmatrix}$ for all $k\geq 0$, where $\delta_{A}=\hat{A}-\bar{A},\delta_{B}=\hat{B}-\bar{B}$. Let $q_{k}=q$ for all $k\geq 0$ in problem \eqref{Weighted pb}. For any fixed $\delta>0$, let $\min{(N_{r},N_{p}})\geq \max\{N_{0},N_{1}\}$, where $N_{0}=8(n+p)+16\log\frac{2T}{\delta}, N_{1} =(4n+2p)\log\frac{T}{\delta}$. Then, with probability at least $1-4\delta$, the least squares solution to problem \eqref{Weighted pb} satisfies
\begin{align}\nonumber
&\max{(\|\bar{A}_{WLS}-\bar{A}\|,\|\bar{B}_{WLS}-\bar{B}\|)}\\\nonumber
\le&\underbrace{\frac{c_{0}(\sqrt{N_{r}}\sigma_{\bar{w}}c_{2}+q\sqrt{N_{p}}\sigma_{\hat{w}}c_{3})}{\lambda_{min}(N_{r}M_{1}+qN_{p}M_{2})}}_\text{Error due to noise}\\
&+\underbrace{q\|\delta_{_{\Theta}}\|\frac{N_{p}c_{1}}{\lambda_{min}(N_{r}M_{1}+qN_{p}M_{2})}}_\text{Error due to difference between true and auxiliary systems},\label{eqn:upper bound on error}
\end{align}
where
\begin{equation*}
\begin{aligned}
&c_{0}=16\sqrt{(2n+p)\log\frac{9T}{\delta}}\\
&c_{1}=9\sum_{k=0}^{T-1}\|\hat{\Sigma}_{k}\|,
\quad c_{2}=\sum_{k=0}^{T-1}\|\bar{\Sigma}_{k}^{\frac{1}{2}}\|,
\quad c_{3}=\sum_{k=0}^{T-1}\|\hat{\Sigma}_{k}^{\frac{1}{2}}\|\\
&M_{1}=\sum_{k=0}^{T-1}\bar{\Sigma}_{k}, 
\quad M_{2}=\sum_{k=0}^{T-1}\hat{\Sigma}_{k}.\\
\end{aligned}
\end{equation*}
\end{coro}

\begin{remark} \textbf{Interpretation of Corollary \ref{Bound LTI} and the selection of weight parameter $q$.}
Recall that $N_{r}$ is the number of rollouts from the true system \eqref{eq:True system}, and $N_{p}$ is the number of rollouts from the auxiliary system \eqref{eq:Perturbed system}. From the upper bound in \eqref{eqn:upper bound on error}, one can observe that when $q$ is fixed, increasing $N_{p}$ reduces the error due to the noise, at the cost of adding a bias that is due to the intrinsic model difference $\delta_{\Theta}$. Choosing the optimal weight $q$ in practice requires an oracle (to know the specific values of the different parameters in \eqref{eqn:upper bound on error}), and one would leverage a cross-validation process to select a good $q$ (see \cite{refaeilzadeh2009cross} for an overview). However, general guidelines can be given based on the upper bound provided by Corollary \ref{Bound LTI} when $N_{p}$ is large: 
\begin{itemize}
    \item When $N_{r}$ is small, we can increase $q$ to reduce the first term in the error bound (such that the error bound is dominated by the second term). This corresponds to the case where we have little data from the true system, and thus there may be a large identification error due to using only that data. In this case, it is worth placing more weight on the data from the auxiliary system, up to the point that the reduction in estimation error due to the larger amount of data is balanced out by the intrinsic differences between the systems.
    \item When $N_{r}$ is large, we can  decrease $q$ to reduce the second term as well, since the first term is already small enough.  This corresponds to the case where we have a large amount of data from the true system, and only need the data from the auxiliary system to slightly improve our estimates (due to the presence of additional data).  In this case, we place a lower weight on the auxiliary data in order to avoid excessive bias due to the difference in dynamics.
    \end{itemize}
The above insights align well with intuition, and are supported by the mathematical analysis provided in this section.  We also illustrate these ideas experimentally in Section \ref{exp}.
\end{remark}

Finally, the following result provides a (sufficient) condition under which using the data from the auxiliary system leads to a smaller error bound (compared to using data only from the true system).

\begin{coro}\label{When to use}
Consider the upper bound provided in Corollary \ref{Bound LTI}. Setting $q$ to be non-zero will result in a smaller upper bound if
\begin{equation}
\begin{aligned}
\frac{\sigma_{\bar{w}}c_{2}}{\sqrt{N_{r}}\lambda_{min}(M_{1})} > \frac{\sigma_{\hat{w}}c_{3}}{\sqrt{N_{p}}\lambda_{min}(M_{2})}+\frac{\|\delta_{\Theta}\|c_{1}}{\lambda_{min}(M_{2})c_{0}}.
\label{eq:condition_to_use_aux_data}
\end{aligned}
\end{equation}
\end{coro}

\begin{proof}
When $q=0$, from Corollary \ref{Bound LTI}, we have
\begin{equation*}
\begin{aligned}
\max{(\|\bar{A}_{WLS}-\bar{A}\|,\|\bar{B}_{WLS}-\bar{B}\|)} \le \frac{c_{0}\sqrt{N_{r}}\sigma_{\bar{w}}c_{2}}{N_{r}\lambda_{min}(M_{1})}.
\end{aligned}
\end{equation*}
Similarly, when $q\neq 0$, from Corollary \ref{Bound LTI}, we have
\begin{multline*}
\max{(\|\bar{A}_{WLS}-\bar{A}\|,\|\bar{B}_{WLS}-\bar{B}\|)}\\
\le \frac{c_{0}(\sqrt{N_{r}}\sigma_{\bar{w}}c_{2}+q\sqrt{N_{p}}\sigma_{\hat{w}}c_{3})}{\lambda_{min}(N_{r}M_{1}+qN_{p}M_{2})}\\
 +q\|\delta_{_{\Theta}}\|\frac{N_{p}c_{1}}{\lambda_{min}(N_{r}M_{1}+qN_{p}M_{2})}.
\end{multline*}

From Weyl's inequality \cite{horn2012matrix}, we have
\begin{equation*}
\begin{aligned}
\lambda_{min}(N_{r}M_{1}+qN_{p}M_{2})\geq N_{r}\lambda_{min}(M_{1})+qN_{p}\lambda_{min}(M_{2}).
\end{aligned}
\end{equation*}
The proof now follows by setting
\begin{equation*}
\begin{aligned}
\frac{c_{0}\sqrt{N_{r}}\sigma_{\bar{w}}c_{2}}{N_{r}\lambda_{min}(M_{1})}&>\frac{c_{0}(\sqrt{N_{r}}\sigma_{\bar{w}}c_{2}+q\sqrt{N_{p}}\sigma_{\hat{w}}c_{3})}{N_{r}\lambda_{min}(M_{1})+qN_{p}\lambda_{min}(M_{2})}\\
&+q\|\delta_{_{\Theta}}\|\frac{N_{p}c_{1}}{N_{r}\lambda_{min}(M_{1})+qN_{p}\lambda_{min}(M_{2})}.
\end{aligned}
\end{equation*}
\end{proof}

\begin{remark} \textbf{Interpretation of Corollary \ref{When to use}.}  We note that the above sufficient condition \eqref{eq:condition_to_use_aux_data} could be conservative, and may not be directly checked in practice due to the unknown parameters. However, we describe the insights provided by this condition. In general, condition \eqref{eq:condition_to_use_aux_data} is more likely to be satisfied if $\sigma_{\hat{w}}$ is small (the auxiliary system is less noisy),  $\|\delta_{\Theta}\|$ is small (the true system and auxiliary system are very similar), or $N_{p}$ is large (one has a lot of samples from the auxiliary system), as this would make the right hand side of \eqref{eq:condition_to_use_aux_data} smaller.  In such cases, the auxiliary samples tend to be more informative. In contrast, condition \eqref{eq:condition_to_use_aux_data} is less likely to be satisfied if $\sigma_{\bar{w}}$ is small and $N_{r}$ is large, i.e., when the data from the true system is not too noisy and one has a lot of samples from that system, the auxiliary samples tend to be less informative. 
\end{remark}

\section{Numerical Experiments to Illustrate Various Scenarios for System Identification from Auxiliary Data} \label{exp}
We now provide some numerical examples of the weighted least squares-based system identification algorithm (Algorithm~\ref{BatchAlgorithm}), when the auxiliary system is LTI. The experiments are performed using the following system matrices:
\begin{align*}
\bar{A}=
\begin{bmatrix}
0.6&0.5&0.4\\
0&0.4&0.3\\
0&0&0.3\\
\end{bmatrix},
\indent 
\bar{B}=
\begin{bmatrix}
1&0.5\\
0.5&1\\
0.5&0.5\\
\end{bmatrix},
\end{align*}
\begin{align*}
\hat{A}=
\begin{bmatrix}
0.7&0.5&0.4\\
0&0.4&0.3\\
0&0&0.3\\
\end{bmatrix},
\indent 
\hat{B}=
\begin{bmatrix}
1.1&0.5\\
0.5&1\\
0.5&0.5\\
\end{bmatrix},
\end{align*}
and $\sigma_{\bar{x}}^{2},\sigma_{\hat{x}}^{2}$, $\sigma_{\bar{u}}^{2},\sigma_{\hat{u}}^{2}$ $\sigma_{\bar{w}}^{2}, \sigma_{\hat{w}}^{2}$ are set to be $1$. The rollout length is set to be $T=2$ for all experiments in the sequel. All results are averaged over $10$ independent runs.

\begin{figure*}
\minipage[t]{0.32\textwidth} 
    \includegraphics[width=\linewidth]{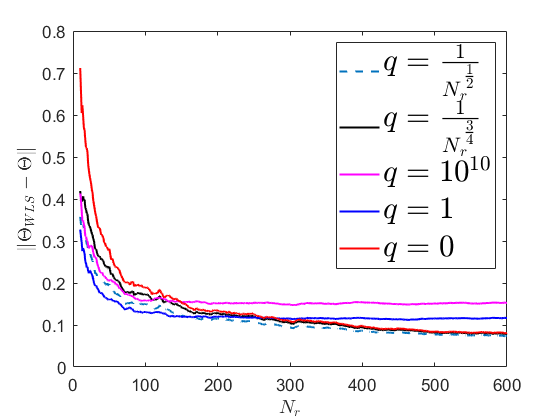}
    \caption{Scenario 1: Both $N_r$ and $N_p$ increase over time ($N_p = 3N_r$)}
    \label{fig:sce1} 
\endminipage \hfill
\minipage[t]{0.32\textwidth}
    \includegraphics[width=\linewidth]{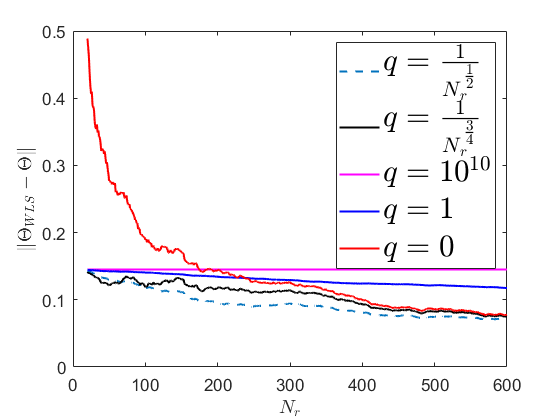}
    \caption{Scenario 2: $N_p$ is fixed, and $N_r$ increases over time}
    \label{fig:sce2} 
\endminipage \hfill
\minipage[t]{0.32\textwidth}
    \includegraphics[width=\linewidth]{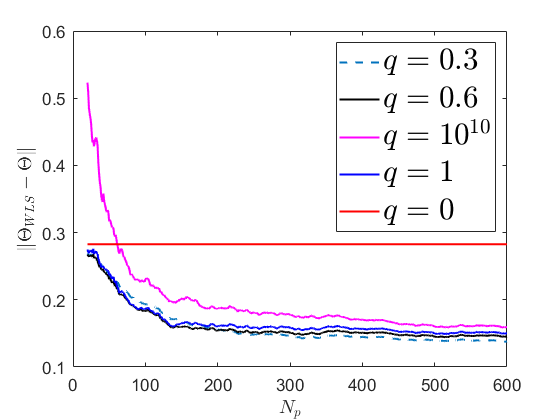}
    \caption{Scenario 3: $N_r$ is fixed, and $N_p$ increases over time}
    \label{fig:sce3} 
\endminipage
\end{figure*}
\subsection{Scenario 1: Both $N_{r}$ and $N_{p}$ are increasing}

For the first experiment, the number of rollouts from the auxiliary system is set to be $N_{p}=3N_{r}$. In practice, one may encounter such a scenario when running experiments to gather data from the true system is time consuming or costly, whereas gathering data from an auxiliary system (such as a simulator) is easier or cheaper.  

In Fig.~\ref{fig:sce1}, we plot the estimation error versus $N_{r}$ for different weight parameters $q$. As expected, setting $q>0$ leads to a smaller estimation error of system matrices when one does not have enough data from the true system (small $N_{r}$). However, the curve for $q=1$ and $q=10^{10}$ (corresponding to treating all samples equally and paying almost no attention to the samples from the true system, respectively) eventually plateau and incur more error than not using the the auxiliary data ($q=0$). This phenomenon matches the theoretical guarantee in Corollary \ref{Bound LTI}. Specifically, when $q$ is a nonzero constant, the upper bound in \eqref{eqn:upper bound on error} will not go to zero as $N_{r}$ increases; furthermore, since both $N_p$ and $N_r$ are increasing in a linear relationship,  there is no need to attach high importance to the auxiliary data when one has enough data from the true system. In contrast, setting $q$ to be diminishing with $N_{r}$ could perform consistently better than $q=0$, even when $N_{r}$ is large. Indeed, one can choose $q=\mathcal{O}(\frac{1}{\sqrt{N_{r}}})$ in the upper bound given by \eqref{eqn:upper bound on error} in~Corollary \ref{Bound LTI}, and show that the upper bound becomes $\mathcal{O}(\frac{1}{\sqrt{N_{r}}})$. Thus, the estimation error tends to zero as $N_r$ increases to infinity.  

{\bf Key Takeaway:} When $N_p$ and $N_r$ are both increasing linearly, having $q$ diminish with $N_r$ helps to reduce the system identification error when $N_r$ is small (by leveraging data from the auxiliary system), and avoids excessive bias from the auxiliary system when $N_r$ is large.

\subsection{Scenario 2: $N_{p}$ is fixed but $N_{r}$ is increasing}
In the second experiment, we fix the number of rollouts from the auxiliary system to be $N_{p}=2400$, and study what happens as the number of rollouts from the true system increases.  One may encounter such a scenario when the system dynamics change due to faults.  In this case, the true system is the one after the fault, and the auxiliary system is the one prior to the fault.  Consequently, the old data from the system (corresponding to the auxiliary system) may not accurately represent the new (true) system dynamics. While one can collect data from the new system dynamics, leveraging the old data might be beneficial in this case.

In Fig.~\ref{fig:sce2}, we plot the estimation error versus $N_{r}$ for different weight parameters $q$. As expected, setting $q>0$ leads to a smaller error during the initial phase when $N_{r}$ is small. This can be confirmed by Corollary \ref{Bound LTI} since the error is essentially the error due to the model difference.  Namely, the auxiliary data helps to build a good initial estimate. When we set the weight to be $q=10^{10}$, we are paying little attention to the samples from the true system, i.e., we are not gaining any new information as we collect more data from the true system. Consequently, the error is almost not changed as $N_{r}$ increases when $q=10^{10}$. As can be observed from Corollary \ref{Bound LTI}, when $N_{p}$ is fixed, consistency can always be achieved as $N_{r}$ increases, using the weights we selected in the experiment. However, when $q$ is set to be too large, it could make the error even larger due to the model difference (or bias) introduced by the auxiliary system. This is captured by Corollary \ref{Bound LTI} since when $q$ is set to be large (such that $qN_{p}$ is large compared to $N_{r}$), even when $N_{r}$ gets large, the upper bound in \eqref{eqn:upper bound on error} is still large due to the effect from the second term in the bound (capturing model difference). 

{\bf Key Takeaway:} When $N_p$ is fixed and large, and $N_{r}$ increases over time, setting $q$ to be nonzero builds a good initial estimate for the true system dynamics when $N_{r}$ is small. Again, having $q$ diminish with $N_r$ helps to reduce the system identification error when $N_r$ is small, and avoids excessive bias from the auxiliary system when $N_r$ is large.

\subsection{Scenario 3: $N_{r}$ is fixed but $N_{p}$ is increasing}
In the last experiment, we fix the number of rollouts from the true system to be $N_{r}=50$.  As discussed earlier, one may encounter such a scenario when  one has only a limited amount of time to gather data from true system. As a result, leveraging information from other systems (e.g., from a reasonably accurate simulator) could be helpful to augment the data.  This is the most subtle case, since $N_{r}$ is fixed and there is no way to guarantee consistency from Corollary \ref{Bound LTI}.

In Fig.~\ref{fig:sce3}, we plot the the estimation error versus $N_{p}$ for different weight parameters $q$.  As it can be seen, setting $q=0$ (not using the auxiliary samples) gives a flat line, which is the error we can achieve purely based on $N_{r}=50$ rollouts from the true system. When $q=10^{10}$, we are essentially identifying the auxiliary system without caring about the true system. In contrast, the results for $q=1, 0.6, 0.3$ suggest that setting a relatively balanced weight makes the error smaller than the two extreme cases ($q=0, 10^{10}$). However, in practice, one may want to utilize a cross-validation process to select a good $q$, when there is not enough prior knowledge about the true system and the auxiliary system. 

{\bf Key Takeaway:} Although consistency cannot be guaranteed when $N_r$ is fixed and small, and $N_{p}$ increases over time, a relatively balanced $q$ could make the error smaller than the extreme cases ($q=0, 10^{10}$).

\section{Conclusion and future work} \label{sec: conclusion}

In this paper, we provided a finite sample analysis of the weighted least squares approach to LTI system identification, when we have access to samples from an auxiliary system that shares similar dynamics. We showed that one can leverage the auxiliary data generated by the similar system to reduce the estimation error due to the noise, at the cost of adding a portion of error that is due to intrinsic differences in the models of the true and auxiliary systems. One limitation of our result is that our bound cannot capture the empirical trend that a longer length of rollout $T$ (when we do not use samples from the auxiliary system) reduces the estimation error. One future direction is to develop a tighter bound by leveraging results from the single trajectory setup. In addition, it is also of interest to consider systems with special structures, such as sparsity.




\bibliographystyle{IEEEtran}
\bibliography{main}
\end{document}